\documentclass[11pt]{article}

\usepackage[paper=a4paper]{geometry}
\usepackage{verbatim}
\usepackage{algpseudocode, algorithm}
\usepackage{amsmath, amsfonts, amsthm}
\usepackage{amssymb}
\usepackage{xspace}
\usepackage[T1]{fontenc}
\usepackage[normalem]{ulem}
\usepackage{braket}

\RequirePackage[l2tabu, orthodox]{nag}

\newcommand{\WARNING}{\text{\textcolor{red}{\textbf{!!!!~WARNING~!!!!}}}}

\renewcommand{\geq}{\geqslant}
\renewcommand{\leq}{\leqslant}
\renewcommand{\epsilon}{\varepsilon}

\newcommand{\alphah}{\widehat{\alpha}}

\newcommand{\st}{\ensuremath{\text{ st.\@ }}}
\newcommand{\Z}{{\mathbb{Z}}}
\newcommand{\C}{{\mathbb{C}}}
\newcommand{\R}{{\mathbb{R}}}
\newcommand{\Q}{{\mathbb{Q}}}
\newcommand{\F}{{\mathbb{F}}}
\newcommand{\N}{{\mathbb{N}}\WARNING} 
\newcommand{\eqdef}{\;\overset{\textup{def}}{=}\;}

\newcommand{\ham}{\ensuremath{\textup{Ham}}}
\newcommand{\concat}{\|}
\newcommand{\upto}{\ldots}
\newcommand{\wildcard}{\ensuremath{\star}\xspace}
\newcommand{\Fam}{\ensuremath{\mathbb{C}_m}} 

\newcommand{\sExactWild}{\textsc{Shift-Exact\textsuperscript{\wildcard}}\xspace}
\newcommand{\DsExactWild}{\ensuremath{d_{\textup{E}}^+}}

\newcommand{\ssExactWild}{\textsc{ShiftScale-Exact\textsuperscript{\wildcard}}\xspace}

\newcommand{\sLtwoWild}{\textsc{Shift-$L_2^\wildcard$}\xspace}
\newcommand{\DsLtwoWild}{\ensuremath{d_2^+}}

\newcommand{\ssLtwoWild}{\textsc{ShiftScale-$L_2^\wildcard$}\xspace}
\newcommand{\DssLtwoWild}{\ensuremath{d_2^{1}}}

\newcommand{\LpolyWild}{\textsc{Poly-$r$-$L_2^\wildcard$}\xspace}
\newcommand{\DLpolyWild}{\ensuremath{d_2^r}}

\newcommand{\sHam}{\textsc{Shift-Ham}\xspace}
\newcommand{\DsHam}{\ensuremath{d_\textup{H}^+}}

\newcommand{\ssHam}{\textsc{ShiftScale-Ham}\xspace}
\newcommand{\DssHam}{\ensuremath{d_\textup{H}^1}}

\newcommand{\skMismatch}{\textsc{Shift-$k$-Mismatch}\xspace}
\newcommand{\DskMismatch}{\ensuremath{d_\textup{M}^+}}

\newcommand{\skDecision}{\textsc{Shift-$k$-Decision}\xspace}
\newcommand{\DskDecision}{\ensuremath{d_\textup{D}^+}}
\newcommand{\dHam}{\DsHam}

\newcommand{\threeSUM}{\textsc{3Sum}\xspace}
\newcommand{\geombase}{\textsc{GeomBase}\xspace}

\newcommand{\freduce}[1]{\ensuremath{\lll_{#1}\!}}

\newcommand{\cb}[1]{\makebox[4.05mm][c]{#1}}
\newcommand{\CB}[1]{\framebox[4.05mm][c]{#1}}
\newcommand{\cbS}[1]{\makebox[4.15mm][l]{#1}}
\newcommand{\cbx}[1]{\cb{$x_{\hspace{-0.0pt}#1}$}}
\newcommand{\cby}[1]{\cb{$y_{\hspace{-0.0pt}#1}$}}
\newcommand{\cbshrink}{\hspace{-0.5pt}}
\newcommand{\cbz}{\cb{0}}
\newcommand{\CBx}[1]{\CB{$x_{#1}$}}

\newcommand{\CBz}{\CB{0}}

\font\tenmsb=msbm10 \font\sevenmsb=msbm7 \font\fivemsb=msbm5

\newfam\msbfam
\textfont\msbfam=\tenmsb \scriptfont\msbfam=\sevenmsb
\scriptscriptfont\msbfam=\fivemsb

\newcommand{\margin}[1]{}

\theoremstyle{plain}
\newtheorem{theorem}{Theorem}[]
\newtheorem{lemma}[theorem]{Lemma}
\newtheorem{corollary}[theorem]{Corollary}
\newtheorem{problem}[theorem]{Problem}

\theoremstyle{definition}
\newtheorem{definition}[theorem]{Definition}

\usepackage[firstinits,style=alphabetic,maxnames=99]{biblatex} 

\renewcommand*{\multicitedelim}{\addcomma\space} 

\bibliography{longnames.bib, bib-latest.bib}

\renewbibmacro{in:}{%
    \ifthenelse{\ifentrytype{article}\OR\ifentrytype{inproceedings}}{}
        {\printtext{\bibstring{in}\intitlepunct}}
}

\AtEveryBibitem{%
    \ifentrytype{inproceedings}{%
        \clearfield{year}%
	\clearfield{publisher}%
	\clearfield{series}%
	\clearfield{editor}%
	\clearfield{address}%
    }{}
}

\makeatletter
\providecommand\bibstyle@faked{}
\providecommand\bibdata@faked{}
\AtBeginDocument{%
\immediate\write\@mainaux{\noexpand\bibstyle@faked}%
\immediate\write\@mainaux{\noexpand\bibdata@faked}}
\makeatother

\title{Pattern Matching under Polynomial Transformation\thanks{~The results in Section~\ref{sec:TIL2}, except for those relating to higher degree polynomials, appeared in preliminary form in ``Self-normalised Distance with Don't Cares'', CPM '07.   Section~\ref{sec:TIk} contains revised and rewritten versions of results from ``Jump-Matching with Errors'', SPIRE '07.}}

\author{%
    \ Ayelet Butman,%
    \thanks{~Department of Computer Science, Holon Institute of Technology, Holon, Israel.}
    \ Peter Clifford,%
    \thanks{~Department of Statistics, University of Oxford, U.K.}
    \ Rapha\"el Clifford,%
    \thanks{~Department of Computer Science, University of Bristol, U.K.}
    \ Markus Jalsenius,%
    \footnotemark[4]~~
    \\Noa Lewenstein,%
    \thanks{~Department of Computer Science, Netanya Academic College, Israel.}
    \ Benny Porat,%
    \thanks{~Department of Computer Science, Bar-Ilan University, Israel.}
    \ Ely Porat%
    \footnotemark[6]
    \ and Benjamin Sach%
    \thanks{~Department of Computer Science, University of Warwick, U.K.}~
}

\date{}

\begin{document}

\maketitle

\begin{abstract}
    We consider a class of pattern matching problems where a normalising transformation is applied at every alignment. Normalised pattern matching plays a key role in fields as diverse as image processing and musical information processing where application specific transformations are often applied to the input.  By considering the class of polynomial transformations of the input, we provide  fast algorithms and the first lower bounds for both new and old problems.

    Given a pattern of length $m$ and a longer text of length $n$ where both are assumed to contain integer values only, we first show $O(n\log{m})$ time algorithms for pattern matching under linear transformations even when wildcard symbols can occur in the input.  We then show how to extend the technique to polynomial transformations of arbitrary degree.  Next we consider the problem of finding the minimum Hamming distance under polynomial transformation. We show that, for any $\epsilon>0$, there cannot exist an $O(nm^{1-\epsilon})$ time algorithm for additive and linear transformations conditional on the hardness of the classic \threeSUM problem. Finally, we consider a version of the Hamming distance problem under additive transformations with a bound $k$ on the maximum distance that need be reported.  We give a deterministic $O(nk \log k)$ time solution which we then improve by careful use of randomisation to $O(n\sqrt{k\log k}\log n)$ time for sufficiently small $k$. Our randomised solution outputs the correct answer at every position with high probability.
\end{abstract}


\section{Introduction}

We consider pattern matching problems where the task is to find the distance between a pattern and every substring of the text of suitable length.  In the class of problems we consider, the values in the pattern can first be transformed so as to minimise this distance. Further, the selection of which transformation to apply, which is possibly distinct for each alignment, forms part of the problem that is to be solved. This class of problems generalises the well known problem of exact matching with wildcards~\cite{CH:2002, Clifford:2007} as well as the set of problems known previously as transposition invariant matching~\cite{MNU:2005}, both of which come from the pattern matching literature. However as we will see, it is considerably broader than both with applications in both image processing and musical information retrieval.

By way of a first motivation for our work, consider a fundamental problem in image processing which is to measure the similarity between a small image segment or template and regions of comparable size within a larger scene. It is well known that the cross-correlation between the
two can be computed efficiently at every position in the larger image
using the fast Fourier transform (FFT).  In practice, images may
differ in a number of ways including being rotated, scaled or affected
by noise.  We consider here the case where the intensity or brightness
of an image occurrence is unknown and where parts of either image
contain \emph{don't care} or \emph{wildcard} pixels, i.e.\ pixels
that are considered to be irrelevant as far as image similarity is
concerned.  As an example, a rectangular image segment may contain a
facial image and the objective is to identify the face in a larger
scene. However, some faces in the larger scene are in shadow and
others in light.  Furthermore, background pixels around the faces may
be considered to be irrelevant for facial recognition and these should
not affect the search algorithm.

In order to overcome the first difficulty of varying intensity within
an image, a standard approach is to compute the \emph{normalised}
distance when comparing a template to part of a larger image.  Thus both template and image are transformed or rescaled
in order to make any matches found more meaningful and to allow comparisons between matches at different positions. Within the image processing literature the accepted method of normalisation is to scale the mean and variance of the template and image segments. We take a slightly different although related approach to normalisation which will allow to us to show a number of natural generalisations.

%

We start by defining measures of distance between a pattern $P$ and text $T$, where $P$ is a string of length $m$ and $T$ is a string of length $n\geq m$, both over the integers.
The squared $L_2$ or Euclidean distance between the pattern and the text at position $i$ is
\begin{equation*}
    \sum_{j=0}^{m-1} \big(P[j] - T[i+j]\big)^2 \,.
\end{equation*}
In this case, for each $i\in\{0,\ldots,n-m\}$, the pattern can be normalised, or fitted as closely as possible to the text, by transforming the input to minimise the distance.

In the case of degree one polynomial transformations, the normalised $L_2$ distance between the pattern and the text at position $i$ can now be written as
\begin{equation*}
    \min_{\alpha, \beta} \sum_{j=0}^{m-1} \big(\alpha + \beta P[j] - T[i+j]\big)^2 \,,
\end{equation*}
where the minimisation is over rational values of $\alpha$ and $\beta$. The minimisation is per alignment of $P$ and $T$, hence the values of $\alpha$ and $\beta$ may (and probably will) differ between the positions $i$.

We also consider the case when the input alphabet is augmented with the special wildcard symbol, denoted~``\wildcard''. A position where either the pattern or text has a wildcard will not contribute to the distance. That is, the minimisation is carried out using the sum of the remaining terms. Details are given in the problem definitions in the next section.


\subsection{Problems and our results}\label{sec:problems}

The words \emph{shift} and \emph{scale} are used to refer to additive and multiplicative transformations of the pattern, respectively. The input to all our problems is a text $T$ of length $n$ and a pattern $P$ of length $m$, and the output is a problem specific distance $d(i)$ between $P$ and $T$ at every position $i\in\{0,\ldots,n-m\}$ of the text. To avoid overloading variable names, we give the distance $d(i)$ a unique name for each problem.

\begin{problem}[\sLtwoWild]
    \label{prob:sLtwoWild}
    Normalised $L_2$ distance under shifts. Wildcards are allowed.
    \begin{equation*}
        \DsLtwoWild(i) \eqdef \min_{\alpha} \sum_{j=0}^{m-1} \big(\alpha + P[j] - T[i+j]\big)^2 \,.
    \end{equation*}
    When either $P[j]=\wildcard$ or $T[i+j]=\wildcard$, the contribution of the pair to the sum $\DsLtwoWild(i)$ is taken to be zero. The minimisation is carried out using the sum of the remaining terms.
\end{problem}

Next we define the normalised $L_2$ distance under shifts and scaling, corresponding to a degree one polynomial transformation of the values of the pattern.

\begin{problem}[\ssLtwoWild]
    Normalised $L_2$ distance under shifts and scaling. Wildcards are allowed.
    \begin{equation*}
        \DssLtwoWild(i) \eqdef \min_{\alpha,\beta} \sum_{j=0}^{m-1} \big(\alpha + \beta P[j] - T[i+j]\big)^2 \,.
    \end{equation*}
    When either $P[j]=\wildcard$ or $T[i+j]=\wildcard$, the contribution of the pair to the sum \DssLtwoWild(i) is taken to be zero. The minimisation is carried out using the sum of the remaining terms.
\end{problem}

We show that both \sLtwoWild and \ssLtwoWild can be solved in $O(n\log{m})$ time by the use of FFTs of integer vectors. Our results are stated in Theorems~\ref{thm:sLtwoWild} and~\ref{thm:ssLtwoWild}. We assume the RAM model of computation throughout in order to be consistent with previous work on matching with wildcards. Further, our techniques also provide $O(n\log{m})$ time solutions (Theorems~\ref{thm:sExactWild} and~\ref{thm:ssExactWild}) to the problems of exact shift matching with wildcards (\sExactWild) and exact shift-scale matching with wildcards (\ssExactWild), formally defined as follows.

\begin{problem}[\sExactWild]
    \label{prob:sExactWild}
    Normalised exact matching under shifts. Wildcards are allowed.
    \begin{equation*}
        \DsExactWild(i) \eqdef
            \begin{cases}
                1, & \text{$\exists\, \alpha$ st.\ $\alpha + P[j]=T[i+j]$ for all $j\in\{0,\ldots,m-1\}$;} \\
                0, & \text{otherwise.}
            \end{cases}
    \end{equation*}
    Every position $j$ where either $P[j]=\wildcard$ or $T[i+j]=\wildcard$ is ignored.
\end{problem}

\begin{problem}[\ssExactWild]
    \label{prob:ssExactWild}
    Normalised exact matching under shifts and scaling. Wildcards are allowed. The problem is defined similarly to \sExactWild, only that we check whether there exist $\alpha$ and $\beta$ \st $\alpha + \beta P[j]=T[i+j]$ for all positions $j$ (except positions where $P[j]=\wildcard$ or $T[i+j]=\wildcard$).
\end{problem}

We will also discuss extensions to pattern transformations under polynomials of higher degree in Section~\ref{sec:TIL2}. In terms of normalised $L_2$ distance we give the following definition.

\begin{problem}[\LpolyWild]
    \label{prob:LpolyWild}
    Normalised $L_2$ distance under degree-$r$ polynomial transformation. Wildcards are allowed. Let $f(x)=\alpha_0 + \alpha_1 x + \alpha_2 x^2 + \cdots + \alpha_r x^r$ be a polynomial of degree~$r$ with $r\geq 1$.
    \begin{equation*}
        \DLpolyWild(i) \eqdef \min_{\alpha_0,\ldots,\alpha_r} \sum_{j=0}^{m-1} \big(f(P[j]) - T[i+j]\big)^2 \,.
    \end{equation*}
    When either $P[j]=\wildcard$ or $T[i+j]=\wildcard$, the contribution of the pair to the sum \DLpolyWild(i) is taken to be zero. The minimisation is carried out using the sum of the remaining terms.
\end{problem}

Note that the problem \ssLtwoWild is the same problem as \LpolyWild with degree $r=1$. We will show that \LpolyWild can be solved in $O(r n\log{m} + r^{w} n)$ time, where $w$ is the exponent for matrix multiplication (e.g., $w\approx 2.38$ when using the Coppersmith-Winograd algorithm).

The second main topic of our work is on normalised pattern matching problems under the Hamming distance. The Hamming distance is perhaps the most commonly considered measure of distance between strings in the field of pattern matching. We therefore define related normalised versions of our pattern matching problems in a similar way to before.

\begin{problem}[\sHam]
    \label{prob:sHam}
    Normalised Hamming distance under shifts. Wildcards are not allowed.
    \begin{equation*}
        \DsHam(i) \;\eqdef\; \min_{\alpha} \big| \Set{j | \alpha + P[j] \neq T[i+j]} \big| \,.
    \end{equation*}
\end{problem}

\begin{problem}[\ssHam]
    \label{prob:ssHam}
    Normalised Hamming distance under shifts and scaling. Wildcards are not allowed.
    \begin{equation*}
        \DssHam(i) \eqdef \min_{\alpha,\beta} \big| \Set{j | \alpha + \beta P[j] \neq T[i+j]} \big| \,.
    \end{equation*}
\end{problem}

%
%

Previously it has been shown that \sHam, sometimes also referred to as transposition invariant matching, can be solved in $O(nm\log{m})$ time~\cite{MNU:2005}.  It has been tempting to believe that it might be possible to improve this time complexity, particularly as there exist algorithms for standard non-normalised pattern matching under the Hamming distance which take $O(n\sqrt{m \log{m}})$ time~\cite{Abrahamson:1987,Kosaraju:1987}.
We show by reductions from the well known \threeSUM problem that for both shift and shift-scale matching under the Hamming distance there cannot exist an $O(nm^{1-\epsilon})$ time algorithm for any $\epsilon>0$ (Theorems~\ref{thm:sHamLower} and~\ref{thm:ssHamLower}).

To circumvent this new conjectured lower bound, we consider as our last problem a shift version of the $k$-mismatch problem. In the $k$-mismatch problem, the Hamming distance is to be reported at every alignment as long as it is at most $k$.  If it is greater than $k$ then the algorithm is only required to report that the Hamming distance is large.  We define the problem as follows.

\begin{problem}[\skMismatch]
    \label{prob:skMismatch}
    Normalised $k$-mismatch under shifts. Wildcards are not allowed.
    \begin{equation*}
        \DskMismatch(i) \eqdef \min(\DsHam(i),\, k+1) \,.
    \end{equation*}
\end{problem}

We first give a simple deterministic $O(nk \log k)$ time solution (Theorem~\ref{thm:detkmis}).   We then consider a decision version of the problem where we output only the locations $i$ where $\DsHam(i) \leq k$ but not the Hamming distance at those locations. The decision version is defined as follows.

\begin{problem}[\skDecision]
    \label{prob:skDecision}
    Normalised $k$-mismatch decision problem under shifts. Wildcards are not allowed.
    \begin{equation*}
        \DskDecision(i) \eqdef
            \begin{cases}
                0, & \text{$\DsHam(i) \leq k$;} \\
                1, & \text{otherwise.}
            \end{cases}
    \end{equation*}
\end{problem}

Using randomisation we show how to solve this problem in $O(cn\sqrt{k\log k}\log n)$ time for the case that $k < \sqrt{m/6}$ (Theorem~\ref{thm:rankmis}). Here $c$ is a constant that can be chosen arbitrarily to fine tune the error probability. Namely, our algorithm outputs the correct answer at every alignment with probability at least $1-1/n^c$.  We therefore succeed in breaking our newly introduced running time barrier provided by the reduction from \threeSUM for a limited range of values of $k$.

\subsection{Related work}

\margin{Transposition inv: Add in transposition invariant previous work}

Combinatorial pattern matching has concerned itself mainly with
strings of symbolic characters where the distance between individual
characters is specified by some convention. For the $k$-mismatch problem, an $O(nk)$ time algorithm was given in 1986 that uses constant time lowest common ancestor queries on the suffix tree of the pattern and text in a technique that has subsequently come to be known as `kangaroo hopping'~\cite{LV:1986a}. Almost $20$~years afterwards, the asymptotic running time was finally improved in~\cite{ALP:2004} to $O(n\sqrt{k\log{k}})$ time by a method based on filtering, the suffix tree (with kangaroo hopping) and FFTs.  In 2002, a deterministic $O(n\log{m})$ time solution for exact matching with wildcards was given by Cole and Hariharan~\cite{CH:2002} and further simplified in~\cite{Clifford:2007}.  In the same paper by Cole and Hariharan, an $O\big(n\log(\max(m,N))\big)$ time algorithm for the exact shift matching problem we consider in Section~\ref{sec:TIL2} was presented. Here $N$ is the largest value in the input.  The approach we take to provide a simpler solution for this problem is similar in spirit to that of~\cite{Clifford:2007}.

There has also been some work in recent years on fast algorithms for distance calculation and approximate matching between numerical strings.  A number of different metrics
have been considered, with for example, $O(n\sqrt{m\log{m}})$ time solutions found for the $L_1$ distance \cite{Atallah:01,CCI:2005,ALPU:2005}
and less-than matching~\cite{Amir:1995} problems and an $O(\delta n \log{m})$ time algorithm for the $\delta$-bounded version of the
$L_{\infty}$ norm first discussed in \cite{CI:2004a} and then improved in \cite{CCI:2005, LP:2005}.

The most closely related work to ours comes under the heading of transposition invariant matching~\cite{LU:2000}.  The original motivation for this problem was within musical information retrieval where musical search is to be performed invariant of pitch level transposition. The transposition invariant distance between two equal lengthed strings $A$ and $B$ is defined to be $\min_{\alpha} d(A+\alpha, B)$, where $A+\alpha$ is the string obtained from $A$ by adding $\alpha$ to every value and the distance $d$ between strings can be variously defined.  Algorithms for transposition invariant Hamming distance, longest common subsequence (LCS) and Levenshtein (edit) distance amongst others were given in~\cite{MNU:2005} whose time complexities are close to the known upper bounds without transposition.   We show, in Section~\ref{sec:3SUM}, lower bounds for the special case of transposition invariant Hamming distance, which we named \sHam. Normalised pattern matching is also of central interest in the image processing literature where normalisation is typically performed by scaling the mean and standard deviation of the template and each suitably sized image segment to be $0$ and $1$, respectively. An asymptotically fast method for performing normalised cross-correlation for template matching, also using FFTs, was given in~\cite{Lewis:1995}.  The methods we give in Section~\ref{sec:TIL2} have some broad similarity to their approach only in the use of FFTs to provide fast solutions. Due to the differences in the definition of normalisation between our work and theirs, the solutions we give are otherwise quite distinct.

As a general class of problems, pattern matching under polynomial transformation is to the best of our knowledge new.  However, if we allow the degree of the polynomial transformation to increase to $m$, then determining for which alignment the normalised distance equals zero is equivalent to the known problem of function matching.  Function matching has a deterministic $O(n|\Sigma_P|\log{m})$ time solution, where $|\Sigma_P|$ is the size of the pattern alphabet, and a faster randomised algorithm which runs in $O(n\log{n})$ time and has failure probability $1/n$~\cite{AALP:2006}.  \margin{Higher degree: Add relationship to our higher degree result when it exists.}

\subsection{Basic notation}

For a string $X$ of length $\ell$, we write $X[i]$ to denote the $i$th character of $X$ such that $X=X[0]\,X[1]\,X[2]\cdots X[\ell-1]$ (the first index is always zero). The $s$-length substring of $X$ starting at position $i$ is denoted $X[i\upto i+s-1]$. For two strings $X$ and $Y$, the notion $X\concat Y$ is used to denote the string formed by concatenating $X$ and $Y$ in that order. All strings in this paper are over the integer alphabet. Therefore, $X[j]Y[j]$ denotes the product of the numerical characters $X[j]$ and $Y[j]$. If strings $X$ and $Y$ are of equal length, we use the notation $X\cdot Y$ for the string with characters $(X\cdot Y)[i]=X[i]Y[i]$. This element-wise arithmetic is used similarly for addition, subtraction, division and power. For example, the $i$th symbol of $X^2 / Y$ is $X[i]^2/Y[i]$. For a real value $k$, the scalar multiplication $kX$ is the string $(kX)[i]=kX[i]$.

The notation $\ham(X,Y)$ will be used to denote the Hamming distance between equal lengthed strings $X$ and $Y$:
\begin{equation*}
   \ham(X,Y) \eqdef \big|\Set{i | X[i] \neq Y[i] } \big| \,.
\end{equation*}

Throughout this paper we use $T$ to denote the text and $P$ for the pattern. We use $n$ to denote the length of $T$ and $m$ for the length of $P$.

Our algorithms in Section~\ref{sec:TIL2} make extensive use of FFTs.  An important property of the FFT is that the cross-correlation, defined as
\begin{equation*}
    (T \otimes P)[i] \eqdef \sum_{j=0}^{m-1}P[j] T[{i+j}] \,,
\end{equation*}
can be calculated accurately and efficiently for all $i\in\{0,\dots,n-m\}$ in $O(n \log m)$ time (see e.g. \cite{Cormen:1990}, Chapter 32). The time complexity is reduced from $O(n\log{n})$ to $O(n\log{m})$ using a standard splitting trick which partitions the text into $2m$ length substrings which overlap each other by $m$ characters. When it is clear from the context we use $\sum$ as an abbreviation for $\sum_{j=0}^{m-1}$.

We use ``\wildcard'' for the single character wildcard symbol. Under arithmetics on strings, as defined above, we may think of a wildcard as having the value zero. This value is, however, inconsequential for our purposes, as all expressions in this paper have the property that whenever a wildcard symbol is involved in some arithmetics, it is multiplied by a zero.

We write $[n]$ to denote the set of integers $\{0 \ldots n-1\}$. We also say that $g(n) \in \tilde{\Omega}(h(n))$ if and only if $g(n) \in \Omega(h(n)/\log^c{n})$ for some constant $c$, i.e $g(n) \in \Omega(h(n))$ up to log factors.

\subsection{Organisation}

The reminder of the paper is organised as follows. In Section~\ref{sec:TIL2} we discuss normalised pattern distance under $L_2$ distance (\sLtwoWild and \ssLtwoWild) and the decision variants (\sExactWild and \ssExactWild). We also show how to extend the methods to transformations of higher degree polynomials (\LpolyWild). Then in Section~\ref{sec:3SUM} we give running time lower bounds for \sHam and \ssHam by reduction from the \threeSUM problem.  In Section~\ref{sec:TIk} we introduce our new deterministic and randomised algorithms for \skMismatch and \skDecision. Finally, we conclude in Section~\ref{sec:discussion} and set out some open problems.

\section{Normalised $L_2$ distance}\label{sec:TIL2}

We give $O(n\log{m})$ time solutions for shift and shift-scale versions of the normalised $L_2$ distance problem with wildcards. We further show that this enables us to solve the exact shift matching and exact shift-scale matching problems in the same time complexity for inputs containing wildcard symbols. Lastly we show how to extend our solutions to normalisation under polynomials of arbitrary degree.

\subsection{Normalised $L_2$ distance under shifts}

In order to handle wildcards, we define two new strings $P'$ and $T'$ obtained from $P$ and $T$, respectively, such that $P'[j]=0$ if $P[j]=\wildcard$, and $P'[j]=1$ otherwise. Similarly, $T'[i]=0$ if $T[i]=\wildcard$, and $T'[i]=1$ otherwise. We can now express the shift normalised $L_2$ distance at position $i$ as
\begin{equation*}
    \DsLtwoWild(i) \,=\, \min_{\alpha} \sum_{j=0}^{m-1} \Big( \big(\alpha + P[j] - T[i+j]\big)^2 \cdot P'[j]\cdot T'[i+j] \Big) \,.
\end{equation*}
Algorithm~\ref{alg:shiftdistance} shows how to compute $\DsLtwoWild(i)$ for all positions $i$. Correctness and running time is given in the following theorem.

\begin{algorithm}[t]
    \caption{Solution to \sLtwoWild.
    \label{alg:shiftdistance}}
    \begin{enumerate}
        \item Construct $P'$ from $P$ such that $P'[j]=0$ if $P[j]=\wildcard$, and $P'[j]=1$ otherwise. Construct $T'$ from $T$ similarly.
        \item Compute the following six cross-correlations:

        \begin{tabular}{l}
            $C_1 \,=\, (T^2 \cdot T') \otimes P'$ \\
            $C_2 \,=\, (T \cdot T') \otimes (P \cdot P')$
        \end{tabular}
        \hfill
        \begin{tabular}{l}
            $C_3 \,=\, T' \otimes (P^2 \cdot P')$ \\
            $C_4 \,=\, (T \cdot T') \otimes P'$
        \end{tabular}
        \hfill
        \begin{tabular}{l}
            $C_5 \,=\, T' \otimes (P \cdot P')$ \\
            $C_6 \,=\, T'\otimes P'$
        \end{tabular}

        \item Return $A = C_1 - 2 C_2 + C_3 - \big( (C_4 - C_5)^2 / C_6 \big)$. We have $\DsLtwoWild(i)=A[i]$. For positions~$i$ where $C_6[i]=0$ we have $\DsLtwoWild(i)=0$.
    \end{enumerate}
    \vspace{-8pt}
\end{algorithm}

\begin{theorem}
    \label{thm:sLtwoWild}
    The shift version of the normalised $L_2$ distance with wildcards problem (\sLtwoWild) can be solved in $O(n\log{m})$ time.\end{theorem}
\begin{proof}
    Consider Algorithm~\ref{alg:shiftdistance}. We first analyse the running time. Step~1 requires only single passes over the input. Similarly, $(P^2 \cdot P')$, $(P \cdot P')$, $(T \cdot T')$ and $(T^2 \cdot T')$ can all be calculated in linear time once $T'$ and $P'$ are known. Using the FFT, the six cross-correlations in Step~2 can be calculated in $O(n\log{m})$ time. The final vector of Step~3 is obtained in linear time. Thus, $O(n\log{m})$ is the overall time complexity of the algorithm.

    To show correctness we consider the minimum value of
    \begin{equation}
        \label{eq:sLtwoMinimise}
        A[i] \,=\, \sum_{j=0}^{m-1} \Big( \big(\alpha + P[j] - T[i+j]\big)^2 \cdot P'[j]\cdot T'[i+j] \Big) \,.
    \end{equation}
    This can be obtained by differentiating with respect to $\alpha$ and obtaining the minimising value. Solving
    \begin{equation*}
        \frac{\partial A[i]}{\partial \alpha} \,=\, 2 \sum_{j=0}^{m-1} \Big( \big(\alpha + P[j] - T[i+j] \big) \cdot P'[j] \cdot T'[i+j] \Big) \,=\, 0
    \end{equation*}
    gives us the value
    \begin{equation*}
        \widehat{\alpha} \,=\, \frac{\sum \Big( \big( T[i+j] - P[j] \big) \cdot P'[j] \cdot T'[i+j] \Big)}{\sum P'[j]\cdot T'[i+j]} \,=\, \frac{\big((T \cdot T') \otimes P'\big) - \big((P \cdot P') \otimes T' \big)}{P'\otimes T'} \,,
    \end{equation*}
    where $\widehat{\alpha}[i]$ is the minimising value at position~$i$. Substituting $\alpha = \widehat{\alpha}$ into Equation~(\ref{eq:sLtwoMinimise}), expanding and collecting terms, we obtain the final answer as
    \begin{equation*}
        A \,=\, C_1 - 2 C_2 + C_3 - \frac{(C_4 - C_5)^2}{C_6} \,,
    \end{equation*}
    where $C_1,\dots,C_6$ are the correlations defined in Algorithm~\ref{alg:shiftdistance}.

    Lastly we observe that when $C_6[i]=(T'\otimes P')[i]=0$ there is a wildcard at every position in the alignment of $P$ and $T$. Here the shift normalised $L_2$ distance is defined to be~0.
\end{proof}

\subsection{Normalised $L_2$ distance under shift-scale}

Similarly to the shift version of normalised $L_2$ distance in the previous section, we can now solve the shift-scale version. The solution is slightly more involved but the running time remains the same. Algorithm~\ref{alg:shift-scaleddistance} sets out the main steps to achieve this and the result is summarised in the following theorem.

\begin{algorithm}[t]
    \caption{Solution to \ssLtwoWild.
    \label{alg:shift-scaleddistance}}
    \begin{enumerate}
        \item Construct $P'$ from $P$ such that $P'[j]=0$ if $P[j]=\wildcard$, and $P'[j]=1$ otherwise. Construct $T'$ from $T$ similarly.
        \item Compute the following six cross-correlations:

        \begin{tabular}{l}
            $C_1 \,=\, (T^2 \cdot T') \otimes P'$ \\
            $C_2 \,=\, (T \cdot T') \otimes (P \cdot P')$
        \end{tabular}
        \hfill
        \begin{tabular}{l}
            $C_3 \,=\, T' \otimes (P^2 \cdot P')$ \\
            $C_4 \,=\, (T \cdot T') \otimes P'$
        \end{tabular}
        \hfill
        \begin{tabular}{l}
            $C_5 \,=\, T' \otimes (P \cdot P')$ \\
            $C_6 \,=\, T'\otimes P'$
        \end{tabular}

        \item Compute

        $B_1 = C_3 \cdot C_4 - C_2 \cdot C_5 \,, \quad
            B_2 = C_3 \cdot C_6 - C_5^2 \,, \quad
            B_3 = C_2 - \dfrac{C_4 \cdot C_5}{C_6} \,, \quad
            B_4 = C_3 - \dfrac{C_5^2}{C_6}$

        and compute $\widehat{\alpha} = B_1/B_2$ and $\widehat{\beta} = B_3/B_4$. At positions $i$ where $C_6[i]=0$, set $\widehat{\alpha}[i]=\widehat{\beta}[i]=0$. At positions $i$ where $B_2[i]=0$ and $C_6[i]\neq 0$, set $\widehat{\alpha}[i]=C_4/C_6$ and $\widehat{\beta}[i]=0$.

        \item Return
            $B = (\widehat{\alpha}^2 \cdot C_6) +
                2 (\widehat{\alpha} \cdot \widehat{\beta} \cdot C_5) -
                2 (\widehat{\alpha} \cdot C_4) +
                (\widehat{\beta}^2 \cdot C_3) -
                2 (\widehat{\beta} \cdot C_2) +
                C_1$.

            We have $\DssLtwoWild(i)=B[i]$.
    \end{enumerate}
    \vspace{-8pt}
\end{algorithm}

\begin{theorem}
    \label{thm:ssLtwoWild}
    The shift-scale version of the normalised $L_2$ distance with wildcards problem (\ssLtwoWild) can be solved in $O(n\log{m})$ time.
\end{theorem}
\begin{proof}
    Consider Algorithm~\ref{alg:shift-scaleddistance}. Notice that the same six correlations as in Algorithm~\ref{alg:shiftdistance} have to be calculated. The additional strings in Step~3 require linear time, as well as producing the output in Step~4. Hence the overall running time is $O(n\log{m})$.

    Similarly to Equation~(\ref{eq:sLtwoMinimise}) we can express the shift-scale version of the normalised $L_2$ distance at position $i$ as
    \begin{equation}
        \label{eq:ssLtwoMinimise}
        B[i] \,=\, \sum_{j=0}^{m-1} \Big( \big(\alpha + \beta P[j] - T[i+j]\big)^2 \cdot P'[j]\cdot T'[i+j] \Big) \,.
    \end{equation}
    By minimising this expression with respect to both $\alpha$ and $\beta$ we get a system of two simultaneous linear equations
    \begin{align*}
        \frac{\partial B[i]}{\partial \alpha} \,&=\, 2 \sum_{j=0}^{m-1} \Big( \big(\alpha + \beta P[j] - T[i+j] \big) \cdot P'[j] \cdot T'[i+j] \Big) \,=\, 0 \,, \\
        \frac{\partial B[i]}{\partial \beta} \,&=\, 2 \sum_{j=0}^{m-1} \Big( \big(\alpha + \beta P[j] - T[i+j] \big) \cdot P[j] \cdot P'[j] \cdot T'[i+j] \Big) \,=\, 0 \,.
    \end{align*}
    By solving this system and using the definitions of $B_1,\dots,B_4$ in Algorithm~\ref{alg:shift-scaleddistance}, we get the minimising values
    \begin{equation*}
        \widehat{\alpha} = \frac{B_1}{B_2} \qquad \text{and} \qquad \widehat{\beta} = \frac{B_3}{B_4} \,.
    \end{equation*}
    For some positions $i$, the solution to the system might not be unique. This happens at alignments $i$ for which every position $i+j$ has a wildcard, hence $C_6[i]=0$. Here we avoid illegal division by zero by simply setting both $\widehat{\alpha}[i]$ and $\widehat{\beta}[i]$ to zero (any value would do). A non-unique solution also occurs at alignments $i$ where all $P[j]$ are identical over every non-wildcard position $i+j$. This is characterised by $B_2[i]=0$. To see this, observe that $ C_5[i]^2\leq C_3[i] C_6[i]$ by Cauchy-Schwarz inequality. Here we set (arbitrarily) $\widehat{\beta}[i]=0$ and therefore obtain the minimising value $\widehat{\alpha}[i]=C_4/C_6$.

    At Stage~4, $\widehat{\alpha}$ and $\widehat{\beta}$ contain the minimising values for $\alpha$ and $\beta$ at every position. We substitute these into Equation~(\ref{eq:ssLtwoMinimise}) and expand. This gives us the expression for $B$.
\end{proof}

\subsection{Exact shift and shift-scale matching with wildcards}

For the exact shift matching problem with wildcards, \sExactWild, a match is said to occur at location $i$ if, for some shift $\alpha$ and  for every position $j$ in the pattern, either $\alpha + P[j]=T[i+j]$ or at least one of $P[j]$ and $T[i+j]$ is the wildcard symbol. Cole and Hariharan~\cite{CH:2002} introduced a new coding for this problem that maps the string elements into~$0$ for wildcards and complex numbers of modulus~$1$ otherwise.  The FFT is then used to find the (complex) cross-correlation between these coded strings, and finally a shift match is declared at location $i$ if the $i$th element of the modulus of the cross-correlation is equal to $(P'\otimes T')[i]$.

Our Algorithm~\ref{alg:shiftdistance} provides a straightforward alternative method for shift matching with wildcards.  It has the advantage of only using simple integer codings.  Since Algorithm~\ref{alg:shiftdistance} finds the minimum $L_2$ distance at location $i$, over all possible shifts, it is only necessary to test whether this distance is zero.  The running time for the test is then $O(n\log{m})$ since it is determined by the running time of Algorithm~\ref{alg:shiftdistance}.

\begin{theorem}
    \label{thm:sExactWild}
    The problem of exact shift matching with wildcards (\sExactWild) can be solved in $O(n\log{m})$ time.
\end{theorem}

The exact shift-scale matching problem with wildcards, \ssExactWild, can be solved similarly by applying Algorithm~\ref{alg:shift-scaleddistance}.

\begin{theorem}
    \label{thm:ssExactWild}
    The problem of exact shift-scale matching with wildcards (\ssExactWild) can be solved in $O(n\log{m})$ time.
\end{theorem}

\subsection{Normalised $L_2$ distance under higher degree transformations}

We can now consider the problem of computing the normalised $L_2$ distance under general polynomial transformations. The problem, which we termed \LpolyWild, was defined in Problem~\ref{prob:LpolyWild}. Recall that we let
\begin{equation*}
    f(x)=\alpha_0 + \alpha_1 x + \alpha_2 x^2 + \cdots + \alpha_r x^r
\end{equation*}
be a polynomial of degree~$r\geq 1$.
%
%
Similarly to the shift and shift-scale versions of the normalised $L_2$ distance we consider the minimum value of
\begin{equation}
    \label{eq:polyLtwoWild}
    D[i] = \sum_{j=0}^{m-1} \Big( \big(f(P[j]) - T[i+j]\big)^2 \cdot P'[j]\cdot T'[i+j] \Big) \,.
\end{equation}
By differentiating with respect to each $\alpha_{k}$ in turn, giving
\begin{equation*}
    \frac{\partial D[i]}{\partial \alpha_k} \,=\, 2 \sum_{j=0}^{m-1} \Big( \big( f(P[j])  - T[i+j] \big) \cdot P[j]^k \cdot P'[j] \cdot T'[i+j] \Big) \,=\, 0 \,,
\end{equation*}
we obtain a system of $r+1$ linear equations in $r+1$ unknowns for each alignment $i$ of the pattern and text. We need to solve these equations and then substitute the minimising $\alpha_k$ values back into Equation~(\ref{eq:polyLtwoWild}) as we did in the proof of Theorem~\ref{thm:ssLtwoWild}. This procedure is captured by the following theorem.

\begin{theorem}
    The normalised $L_2$ distance problem with wildcards under polynomial transformations of degree $r$ (\LpolyWild) can be solved in  $O(rn\log{m} +   r^{2.38}n)$ time.
\end{theorem}
\begin{proof}
    To compute the coefficients for the first linear equation for $\alpha_0$ we need to perform $O(r)$ cross-correlations.  However, for each subsequent equation for $\alpha_1 \ldots \alpha_r$ we only need to perform a constant number of new cross-correlations.  Therefore the total number of cross-correlations is $O(r)$ to give the coefficients of all the equations, taking $O(r n\log{m})$ time overall.  The time to solve the systems of $O(r)$ equations in $O(r)$ unknowns is $O(r^w)$ per alignment $i$, where $w$ is the exponent for matrix multiplication.  This gives $O(nr^w)$ time or $O(nr^{2.38})$ using the algorithm of Coppersmith and Winograd~\cite{CW:1990}.

    Once the equations have been solved, and the minimising values of $\alpha_k$ calculated, they are then substituted into the expression for $D$ in Equation~(\ref{eq:polyLtwoWild}).   To calculate the final values $D[i]$ we require $O(r)$ cross-correlations to be computed as well as $O(r^2)$ products of vectors of length $m$.  The overall time complexity is therefore $O(rn\log{m} +   r^{2.38}n + r^2 m)$.
\end{proof}

This method is of particular relevance for low degree polynomials, or at least polynomials whose degree is less than the number of distinct values in the pattern.  However, if the degree $r$ is greater than the number of distinct values in the pattern, then there exists a suitable polynomial $f$ for any  mapping we should choose.  This gives us a straightforward $O(nm)$ time solution by considering each position of the pattern in the text independently and ignoring any values aligned with wildcards in either the pattern or text. For each such position we need only set $f(P[j])$ to be the mean of the values in the text that align with a value equal to $P[j]$ in the pattern.

\section{Lower bounds for Hamming distance}\label{sec:3SUM}

In this section we will show that no $O(nm^{1-\epsilon})$ time algorithm can exist for neither \sHam or \ssHam conditional on the hardness of the classic \threeSUM problem. One formulation of the \threeSUM problem is given below.

\begin{definition}[\threeSUM]
    \label{def:threeSUM}
    Given a set of $s$ positive integers, determine whether there are three elements $a,b,c$ in the set such that $a+b=c$.
\end{definition}

The \threeSUM problem can be solved in $O(s^2)$ time and it is a long standing conjecture that this is essentially the best possible. The problem has been extensively discussed in the literature, where Gajentaan and Overmars~\cite{GO:1995} were the first to introduce the concept of \threeSUM-hardness (see definition below) to show that a wide range of problems in computational geometry are at least as hard as the \threeSUM problem. One example is the \geombase problem, defined below, which we will use in one of our reductions in this section. See~\cite{King:2004} for a survey of problems from computational geometry whose hardness relies on that of \threeSUM.

\begin{definition}[\geombase]
    \label{def:geombase}
    Given a set of $s$ points with integer coordinates on three horizontal lines $y=0$, $y=1$ and $y=2$, determine whether there exists a non-horizontal line containing three of the points.
\end{definition}

Although an $\tilde{\Omega}(s^2)$ lower bound for \threeSUM is only conjectured, it has been shown that under certain restricted models of computation, $\Omega(s^2)$ is a true lower bound (see~\cite{ES:1995,Erickson:1999,Erickson:convex:1999}). Under models that allow more direct manipulation of numbers instead of just real arithmetic, such as the word-RAM model, an almost $\log^2{s}$ factor improvement to the standard $O(s^2)$ solution has been shown to be possible under the Las Vegas model of randomisation (see~\cite{BDP:2005}). Nevertheless, a \threeSUM-hardness result for a problem is a strong indication that finding an $O(s^{2-\epsilon})$ time solution is going to be a challenging task.

Before we show that \sHam and \ssHam are both \threeSUM-hard, we provide a brief but formal discussion about reductions and define \threeSUM-hardness.

\subsection{\threeSUM reductions}

Following the definitions of~\cite{GO:1995} where \threeSUM-hardness was first introduced, we say that a problem $A$ is \emph{$g(s)$-solvable using a problem $B$} if and only if every instance of $A$ of size $s$ can be solved using a constant number of instances of $B$ of at most $O(s)$ size and $O(g(s))$ additional time. We denote this as $A\freduce{g(s)}B$. When $g(s)$ is sufficiently small, lower bounds for $A$ carry over to $B$. A problem $B$ is \emph{\threeSUM-hard} if $\threeSUM\freduce{g(s)} B$ and $g(s)=o(s^{2-\epsilon})$ for some constant $\epsilon>0$. In the definition of \threeSUM-hardness of~\cite{GO:1995}, the requirement was that $g(s)=o(s^2)$, however, to scale with more powerful models of computation, we require that $g(s)=o(s^{2-\epsilon})$. If $A\freduce{g(s)}B$ and $B\freduce{g(s)}A$ then we say that $A$ and $B$ are \emph{$g(s)$-equivalent}.

In the following section we will show that $\threeSUM\freduce{s\log s}\sHam$ where the instance size of \sHam is a text of length $n=5s$ and a pattern of length $m=3s$.

In the literature there are a variety of definitions of the \threeSUM problem. They differ only slightly in their formulations and are all equivalent. One common definition, used as the ``base problem'' in~\cite{GO:1995}, is formulated as follows. Given a set of $s$ integers, determine whether there are three elements $a,b,c$ in the set such that $a+b+c=0$. Without too much work, one can show that this definition is $O(s)$-equivalent with Definition~\ref{def:threeSUM} of \threeSUM above (small modifications of the proof of Theorem~3.1 in~\cite{GO:1995} can be used to prove this). Further, it was shown in~\cite{GO:1995} that \geombase is $O(s)$-equivalent to \threeSUM.

\subsection{\threeSUM-hardness of \sHam}

In this section we show that \sHam is \threeSUM-hard.

\begin{lemma}
    \label{lem:sHamHard}
    $\threeSUM\freduce{s\log s}\sHam$ where the instance size of \sHam is a text of length $5s$ and a pattern of length $3s$.
\end{lemma}
\begin{proof}
    Let the set $S$ be an instance of \threeSUM of size $s=|S|$. First we sort all elements of $S$ so that $S=\{x_1,\dots,x_s\}$ where $x_1< x_2<\cdots< x_s$. Let $y_1=2x_s+1$ and for $i\in \{2,\dots,2s\}$, let $y_i=y_{i-1}+1$. Thus, $x_s<y_1<\cdots< y_{2s}$. We define the following $s$-length strings over the alphabet $\{x_1,\dots,x_s\}\cup\{y_1,\dots,y_{2s}\}\cup\{0\}$.
    \begin{align*}
        S_0 &= 0 \, 0 \, \cdots \, 0 \quad \textup{($s$ zeros)} &\qquad S_3 &= y_{s+1} \, y_{s+2} \cdots \, y_{2s} \\
        S_1 &= x_1 \, x_2 \cdots \, x_s &\qquad S_4 &= x_s \, x_{s-1} \cdots \, x_1 \\
        S_2 &= y_1 \, y_2 \cdots \, y_s &\qquad ~
    \end{align*}
    We now construct an instance of \sHam specified by
    \begin{align*}
        T = S_0 \concat S_1 \concat S_2 \concat S_1 \concat S_3 \qquad \textup{and} \qquad
        P = S_4 \concat S_0 \concat S_0\,.
    \end{align*}
    The text $T$ has length $n=5s$ and the pattern $P$ has length $m=3s$. First we show that if there are elements $a,b,c\in S$ such that $a+b=c$ then there is a position $i$ such that the shift-normalised Hamming distance between $P$ and $T[i\upto i+m-1]$ is at most $m-2$. We will then show that if no such three elements exist then the shift-normalised Hamming distance between $P$ and every $m$-length substring of $T$ is strictly greater than $m-2$.

    As an illustrative example, suppose that $S$ contains seven elements and suppose that $x_4+x_3=x_6$. Consider the alignment of $P$ and $T$ where $x_4$ in $P$ is aligned with $x_6$ in $T$:
    \begin{center}
        \footnotesize
        \cbS{T:}%
        \cbz\cbz\cbz\cbz\cbz\cbz\cbz%
        \cbx{1}\cbx{2}\cbx{3}\cbx{4}\cbx{5}\CBx{6}\cbx{7}%
        \cby{1}\cby{2}\cby{3}\cby{4}\cby{5}\cby{6}\cby{7}%
        \cbx{1}\cbx{2}\CBx{3}\cbx{4}\cbx{5}\cbx{6}\cbx{7}%
        \cby{8}\cby{9}\cby{\cbshrink 1\!0}\cby{\cbshrink 1\!1}\cby{\cbshrink 1\!2}\cby{\cbshrink 1\!3}\cby{\cbshrink 1\!4}\\
        \cbS{P:}%
        \cb{~}\cb{~}\cb{~}\cb{~}\cb{~}\cb{~}\cb{~}%
        \cb{~}\cb{~}%
        \cbx{7}\cbx{6}\cbx{5}\CBx{4}\cbx{3}\cbx{2}\cbx{1}%
        \cbz\cbz\cbz\cbz\cbz\cbz\cbz%
        \CBz\cbz\cbz\cbz\cbz\cbz\cbz%
        \cb{~}\cb{~}\cb{~}\cb{~}\cb{~}
    \end{center}
    We observe that shifting the pattern by $x_3$ will induce two matches, marked with the squares above. Thus, the shift-normalised Hamming distance is at most $m-2$ (in fact, it is exactly $m-2$). It should be easy to see how this generalises to any size of $S$ and any three elements $a,b,c\in S$ such that $a+b=c$. Namely, the alignment in which $a$ is aligned with $c$ has Hamming distance at most $m-2$ since there must also be a match at the position where 0 aligned with $b$. The construction of $P$ and $T$ ensures that there is always an alignment that captures these matches.

    Now suppose there are no elements $a,b,c\in S$ such that $a+b=c$. Consider a fixed alignment of $P$ and $T$. We will show that there can be at most one match under any shift. By construction of $P$ and $T$, the zeros in $P$ are all aligned with distinct symbols in $T$. Hence for any shift, at most one of these zeros can be involved in a match. The non-zero symbols of $P$ (i.e.,~the $s$-length prefix of $P$) appear in strictly decreasing order and are aligned with an $s$-length substring of $T$ whose elements appear in non-decreasing order. Therefore, under any shift, at most one of the non-zero symbols in $P$ can be involved in a match. It remains to show that there is no shift such that both a zero and a non-zero symbol in $P$ are simultaneously involved in a match. First, we observe that if there is a match between a $0$ in $P$ and some $y_j$ in $T$ then there can be no other match as every non-zero symbol in $P$ is aligned with a value that is less than $y_j$. Suppose therefore that there is a match between a $0$ in $P$ and some $x_j$ in $T$ (i.e.,~the shift is $x_j$). We need to consider three possible cases: there is also a match that involves some $x_k$ in $P$ aligned with either (i)~a~$0$ in $T$, (ii)~some $y_\ell$ in $T$ or (iii)~some $x_\ell$ in $T$. In case~(i) the shift must be negative, hence is not compatible with the shift $x_j$. In case~(ii) we can see that the shift must be greater that $x_s$ (the largest elements in the set $S$), hence is not compatible with the shift $x_j$. In case~(iii) we have that $x_k+x_j=x_\ell$, which contradicts the assumption that there are no elements $a,b,c\in S$ such that $a+b=c$. Thus, the shift-normalised Hamming distance is at least $m-1$ for any alignment of $P$ and $T$.

    Finally, we observe that the most time consuming part of the reduction is the sorting of $S$ which could take $O(s\log s)$ time. This concludes the proof.
\end{proof}

\begin{theorem}
    \label{thm:sHamLower}
    \sHam has no $O(nm^{1-\epsilon})$ time algorithm, for any $\epsilon>0$, conditional on the hardness of the \threeSUM problem.
\end{theorem}
\begin{proof}
    Given a \threeSUM instance of size $s$, by Lemma~\ref{lem:sHamHard} we construct a \sHam instance of size $n=5s$ and $m=3s$ in $O(s\log s)$ time. If \sHam has an $O(nm^{1-\epsilon})$ time algorithm then \threeSUM can be solved in $O(s^{2-\epsilon})$ time.
\end{proof}


Notice that \sHam has an $O(nm\log m)$ time solution~\cite{MNU:2005}. See Section~\ref{sec:unbounded} for details.

\subsection{\threeSUM-hardness of \ssHam}

In this section we show that \ssHam is \threeSUM-hard.

\begin{lemma}
    \label{lem:ssHamHard}
    $\threeSUM\freduce{s}\ssHam$ where the instance size of \sHam is a pattern and text of length $s$ each.
\end{lemma}
\begin{proof}
    We reduce from the \geombase problem which is $O(s)$-equivalent to \threeSUM. Before we describe the reduction we adopt a formulation of the \geombase problem that differs slightly in notation. Instead of insisting on the points being on the horizontal lines $y=0$, $y=1$ and $y=2$, we assume that the points are on the vertical lines $x=0$, $x=1$ and $x=2$ and we want to determine whether there is a (non-vertical) line containing three points. Under this formulation, let $S$ be an instance of \geombase that contains the integer points $(x_1,y_1), (x_2,y_2), \dots, (x_s,y_s)$, where every $x_j\in\{0,1,2\}$.

    We construct an instance of \ssHam that is specified by the text $T=y_1\,y_2\cdots y_s$ and the pattern $P=x_1\,x_2\cdots x_s$, both of length $s$. It should now be clear that \ssHam returns the shift-and-scale normalised Hamming distance $s-3$ (for the only alignment of $P$ and $T$) if and only if there are two values $\alpha$ and $\beta$ such that $\beta x_j + \alpha = y_j$ for three distinct positions $j$, which is equivalent to fitting a line through three points. Note that we minimise $\alpha$ and $\beta$ over the rationals, and any line going through three points is indeed specified by rational values of $\alpha$ and $\beta$. Since the reduction takes linear time, we have proved the lemma.
\end{proof}

\begin{theorem}
    \label{thm:ssHamLower}
    \ssHam has no $O(nm^{1-\epsilon})$ algorithm, for any $\epsilon>0$, conditional on the hardness of the \threeSUM problem.
\end{theorem}
\begin{proof}
    Given a \threeSUM instance of size $s$, by Lemma~\ref{lem:ssHamHard} we construct a \ssHam instance of size $n=m=s$ in $O(s)$ time. If \ssHam has an $O(nm^{1-\epsilon})$ algorithm then \threeSUM can be solved in $O(s^{2-\epsilon})$ time.
\end{proof}


\section{Normalised $k$-mismatch under shifts}\label{sec:TIk}

In this section we consider two versions of the normalised $k$-mismatch problem under shifts, defined as Problems~\ref{prob:skMismatch} and~\ref{prob:skDecision} in the introduction. Both problems are parameterised by an integer $k$. In the first problem, \skMismatch, the output is the shift-normalised Hamming distance between $P$ and $T$ at every position for which the distance is $k$ or less. Where the distance is larger than $k$, only $k+1$ is outputted. Recall from the introduction that the shift-normalised Hamming distance between $P$ and $T[i\upto i+m-1]$ is denoted $\DsHam(i)$ and defined by
\begin{equation*}
        \dHam(i) \eqdef \min_{\alpha} \big| \Set{j | \alpha + P[j] \neq T[i+j]} \big| \,.
\end{equation*}
In Section~\ref{sec:detkmis} we give a deterministic algorithm that solves \skMismatch in $O(nk \log k)$ time.

In Section~\ref{sec:randkmis} we consider the the second version of shift-normalised $k$-mismatch, \skDecision, which unlike the previous problem only indicates with yes or no whether the shift-normalised Hamming distance is $k$ or less. We give a randomised solution to this decision problem with the improved running time $O(cn \sqrt{k \log k} \log n)$. The parameter $c$ is a constant that can be chosen arbitrarily to fine tune the error probability. Namely, the probability that our algorithm outputs the correct answer at every alignment is at least $(1-1/n^c)$. The errors are one-sided such the algorithm will never miss reporting an alignment for which the shift-normalised Hamming distance is indeed $k$ or less. Our algorithm requires that $k<\sqrt{m/6}$, hence it is suited to situations where the locations of text substrings similar to the pattern are required but the distances themselves are not needed.

\subsection{The unbounded case}\label{sec:unbounded}

In~\cite{MNU:2005}, M\"{a}kinen, Navarro and Ukkonen gave an $O(nm \log m)$ time algorithm for the shift-normalised Hamming distance problem, \sHam, which by definition solves the bounded, $k$-mismatch variant in  $O(nm \log m)$ time also. We briefly recap their method by way of an introduction. First observe that the maximum number of matches for any alignment is exactly
\begin{equation*}
    m - \dHam(i)  =  \max_\alpha \Set{ j | T[{i+j}]-P[j]=\alpha }\,.
\end{equation*}
For each alignment $i$, this value can be obtained by creating an $m$-length array $A_i$, which we refer to as the \emph{shift array}, defined by
\begin{equation}
    \label{eq:shiftarray}
    A_i[j]= T[{i+j}]-P[j]
\end{equation}
for all $j\in [m]$. This shift array is then sorted to find the most frequent value, which is the $\alpha$ that minimises $\dHam(i)$. The number of times it occurs is $m-\dHam(i)$. Computing this requires $O(m \log m)$ time per alignment and hence $O(nm \log m)$ time overall. In the next section we will reconsider $A_i$ and demonstrate that it can be run-length encoded in $O(k)$ runs whenever $\dHam\leq k$.

\subsection{A deterministic solution}\label{sec:detkmis}

The deterministic algorithm makes use of the notion of difference strings which were introduced in~\cite{LU:2000} and are defined as follows.

\begin{definition}
    \label{dfn:diffstr}
    Let $S$ be a string of length $s$. The \emph{difference string} of $S$, denoted $S_{\delta}$, is defined by
    \begin{equation*}
        S_\delta[j]=S[j+1]-S[j]
    \end{equation*}
    for all $j\in [s-1]$. The length of $S_{\delta}$ is $s-1$.
\end{definition}

We will also make use of a generalisation of the difference string when we present our randomised algorithm in Section~\ref{sec:randkmis}. The core of our deterministic shift-normalised $k$-mismatch algorithm is the relationship between the number of mismatches between $P_\delta$ and $T_{\delta}[i \upto i+m-2]$ and the value of $\dHam(i)$. We begin in Lemma~\ref{lem:TI-up} below by showing that if $\dHam(i)$ is small then the number of mismatches between the difference strings $P_\delta$ and $T_\delta$ is also small. In~\cite{MNU:2005} a related result was used to reduce the shift-normalised exact matching problem to the conventional exact matching problem. Specifically, they observed that in the special case that $k=0$, the implication becomes an equivalence, i.e., $\dHam(i)=0$ if and only if $P_\delta=T_{\delta}[i\upto i+m-2]$. Unfortunately, this is not the case in general.

\begin{lemma}
    \label{lem:TI-up}
    Let $P$ be a pattern and $T$ a text. For all $i$,
    \begin{equation*}
       \dHam(i)\leq k \,\implies\, \ham\big(P_\delta,\, T_{\delta}[i\upto (i+m-2)]\big) \leq 2k \,.
    \end{equation*}
\end{lemma}
\begin{proof}
    Let $i$ be such that $\dHam(i)\leq k$ and therefore there exists an $\alpha$ such that for at most $k$ distinct position $j\in [m]$ we have that $P[j]+\alpha \neq T[i+j]$. Further, at most $2k$ distinct positions $j\in [m-1]$ have either $P[j]+\alpha \neq T[i+j]$ or $P[j+1]+\alpha \neq T[i+j+1]$. This implies that there are at least $(m-1)-2k$ distinct positions $j\in [m-1]$ such that $P[j]+\alpha = T[i+j]$ and $P[j+1]+\alpha = T[i+j+1]$. By rearranging these equations, for any such $j$ we have that $P[j+1]-P[j]=T[i+j+1]-T[i+j]$ and hence by Definition~\ref{dfn:diffstr}, $P_\delta[j]=T_\delta[1+j]$. As required there are at most $2k$ mismatches (recall $|P_\delta|=m-1$).
\end{proof}

Lemma~\ref{lem:TI-up} suggests the following strategy. First we find the leftmost up to $2k+1$ mismatches between $P_\delta$ and $T_{\delta}[i\upto i+m-2]$ at each alignment $i$. By Lemma~\ref{lem:TI-up} we can disregard any alignments with more than $2k$ mismatches. Finally we use the locations of these mismatches to infer $\dHam(i)$ at the remaining alignments.

The first step can be done using any $k$-mismatch (strictly $2k$-mismatch) algorithm which returns the locations of the mismatches. The well-known `kangaroo' method of~\cite{LV:1986a} achieves this in optimal $O(nk)$ time. The method is so named as it uses longest common extensions to `hop' between mismatches in constant time. The discarding phase is trivial and therefore we only focus on computing $\dHam(i)$ from the locations of the (at most $2k$) mismatches between $P_\delta$ and $T_{\delta}[i\upto (i+m-2)]$, where $i$ is an arbitrary non-discarded alignment.

Recall from Section~\ref{sec:unbounded} the definition of the shift array $A_i$ in Equation~(\ref{eq:shiftarray}), and recall that the value of $m-\dHam(i)$ is the number of occurrences of the most frequent entry in $A_i$. We will now use the locations of the mismatches between $P_\delta$ and $T_{\delta}[i\upto i+m-2]$ to obtain a run-length encoded version of $A_i$ containing $O(k)$ runs. The key property we require is given in Lemma~\ref{lem:run} which states that a matching substring in $P_\delta$ and $T_{\delta}[i\upto i+m-2]$ corresponds to a run (a substring of equal values) in $A_i$. This immediately implies that $A_i$ can be decomposed into at most $4k+1$ runs. Specifically, one run of length~$1$ for each mismatch and an additional run for each stretch between mismatches.

\begin{lemma}
    \label{lem:run}
    If $P_\delta[\ell \upto r] = T_\delta[(i+\ell) \upto (i+r)]$ then $A_i[j]=A_i[\ell]$ for all $j \in\{\ell,\ldots,r\}$.
\end{lemma}
\begin{proof}
    Suppose that $P_\delta[\ell \upto r] = T_\delta[(i+\ell)\upto (i+r)]$. We proceed by induction on $j \in\{\ell,\ldots,r\}$. The base case $j=\ell$ is tautologically true. For the inductive step, let $j \in\{\ell+1,\ldots,r\}$. By the inductive hypothesis, we have that $A_i[j-1]=T[i+j-1]-P[j-1]=A_i[\ell]$. As $P_\delta[j-1]=T_\delta[i+j-1]$, by Definition~\ref{dfn:diffstr} (and rearranging the equation), we have $A_i[j] = T[i+j]-P[j]=T[i+j-1]-P[j-1]=A_i[\ell]$.
\end{proof}

In Section~\ref{sec:unbounded} we discussed that the value of $\dHam(i)$ equals $m-\max_\alpha \big|\Set{ j | A_i[j]=\alpha } \big|$, which could be found by sorting and scanning $A_i$ in $O(m \log m)$ time. However, we now have $A_i$ in run-length encoded form (with $O(k)$ runs), therefore the time taken to find $\dHam(i)$ is reduced to $O(k \log k)$. Over all alignments, this gives $O(nk \log k)$ time as desired.

We can now give an overview of our deterministic algorithm for \skMismatch. The steps are described in Algorithm~\ref{alg:deterministic} and the overall running time is given in Theorem~\ref{thm:detkmis} below.

\begin{algorithm}[t]
    \caption{Overview of deterministic solution to \skMismatch.
        \label{alg:deterministic}}
    \begin{enumerate}
        \item Compute the difference strings $P_\delta$ and $T_\delta$ by scanning $P$ and $T$.
        \item Run a $2k$-mismatch algorithm on $P_\delta$ and $T_\delta$ in order to find all alignments where the number of mismatches is at most $2k$. The $2k$-mismatch algorithm must also return the locations of the mismatches at any alignment where there are at most $2k$ mismatches.
        \item Discard all alignments with more than $2k$ mismatches.
        \item For each undiscarded alignment $i$, decompose $A_i$ into at most $4k+1$ runs (substrings with a common value). The start and end points of the runs are determined by scanning the locations of the  mismatches between $P_\delta$ and $T_\delta[i \ldots i+m-1]$.
        \item Sort the runs in $A_i$ by value in order to find the most frequent entry $\alpha$ in $A_i$. Then output $m-\big|\Set{ j | A_i[j]=\alpha }\big|$, which is the value $\dHam(i)$.
    \end{enumerate}
    \vspace{-8pt}
\end{algorithm}

\begin{theorem}
    \label{thm:detkmis}
    The shift-normalised $k$-mismatch problem (\skMismatch) can be solved deterministically in $O(n k \log k)$ time.
\end{theorem}
\begin{proof}
    The solution is outlined in Algorithm~\ref{alg:deterministic}. Correctness follows directly from the discussion in this section. The time complexity of the five steps is as follows. By inspection of the definition, the difference strings computed in Step~1 require $O(n)$ time. Step~2 uses a $2k$-mismatch algorithm as a black box and can be performed in $O(nk)$ time by using for example the algorithm in~\cite{LV:1986a}. Step~3 makes a single pass of the output of the $2k$-mismatch algorithm in $O(n)$ time. Step~4 constructs a run length encoded version of $A_i$ for each undiscarded~$i$. This requires scanning the $O(k)$ mismatches at each undiscarded alignment. Therefore Step~4 takes $O(nk)$ time. Step~5 scans and sorts each $A_i$ which takes $O(k \log k)$ time per alignment as $A_i$ is encoded by $O(k)$ runs. Overall the algorithm requires $O(nk \log k)$ time as claimed.
\end{proof}

\subsection{An improved, randomised solution} \label{sec:randkmis}

We now present an improved solution to the shift-normalised $k$-mismatch problem which runs in $O(cn \sqrt{k \log k} \log n)$ time. The improved algorithm is for the case that $k<\sqrt{m/6}$ and is randomised. The errors are one-sided (false-positives) and it outputs the correct answer at all alignments with probability at least $1-1/n^c$ for any constant $c$. For each position $i$, the algorithm gives a yes/no answer to the question ``is $\dHam(i) \leq k$?''. The algorithm does not output the actual distance $\dHam(i)$. Throughout this section, we use $T_i$ as shorthand for $T[i \upto i+m-1]$.

In Section~\ref{sec:detkmis} our deterministic algorithm made use of the locations of mismatches in the difference strings $P_\delta$ and $T_\delta[i \ldots i+m-1]$. Recall that the difference string $S_\delta$ was defined to give the differences between consecutive positions in a string~$S$. That is, $S_\delta[j]=S[j+1]-S[j]$ for all~$j$. A key observation was that $P_\delta[j]=T_\delta [i+j]$ if and only if $P[j]-T[i+j]=P[j+1]-T[i+j+1]=-\alpha$, i.e., the positions of $P[j]$ and $P[j+1]$ require the same shift $\alpha$ to match. However, there is no reason to consider only consecutive differences. In fact, as we will see, one may consider differences under any arbitrary permutation of the position set. This notion is formalised as follows.

\begin{definition}
    \label{dfn:pdiffstr}
    Let $S$ be a string of length $s$ and $\pi : [m] \to [m]$ be a permutation. The \emph{permuted difference string} of $S$ under $\pi$, denoted $S_\pi$, is defined by
    \begin{equation*}
        S_\pi[j] = S[\pi(j)] - S[j]
    \end{equation*}
    for all $j\in [s]$. The length of $S_\pi$ is $s$.
\end{definition}

Note that the permuted difference string  $S_\pi$ has length $|S|$ in contrast to the difference string $S_\delta$ of Definition~\ref{dfn:diffstr} which has length $|S|-1$.

The central idea of our improved algorithm is to use the value of $\ham(P_\pi,(T_i)_{\pi})$ to directly determine whether $\dHam(i) \leq k$ at each alignment~$i$. In Definition~\ref{dfn:TI-ktight} below we introduce the notion of a permutation being \emph{$k$-tight} for some $P,T_i$. Intuitively, $\pi$ is $k$-tight for $P,T_i$ if we can infer directly from $\ham(P_\pi,(T_i)_{\pi})$ whether $\dHam(i) \leq k$.

\begin{definition}
    \label{dfn:TI-ktight}
    Let $\pi$ be a permutation, $P$ a pattern and $T_i$ a text substring. We say that $\pi$ is \emph{$k$-tight} for $P,T_i$ if
    \begin{equation*}
        \dHam(i)\leq k \,\iff\, \ham(P_\pi,(T_i)_{\pi}) \leq 2k\,.
    \end{equation*}
\end{definition}

It would of course be highly desirable to find a permutation $\pi$ which is $k$-tight for all $P,T_i$ and any $k$. However, we will see that this is in general not possible. 

We begin by showing that any $\pi$ has the property that $\dHam(i) \leq k$  implies that $\ham(P_\pi,(T_i)_{\pi}) \leq 2k$ for all $P,T_i$. To do so we first prove a general lemma which will also be useful later. Lemma~\ref{lem:TI-permup} then gives the desired property and is a generalisation of Lemma~\ref{lem:TI-up} to arbitrary permutations.

\begin{lemma}
    \label{lem:TI-kmisiff}
    Let $\pi$ be a permutation, $P$ a pattern and $T_i$ a text substring. For all $j\in [m]$,
    \begin{equation*}
        P[j]-T_i[j] = P[\pi(j)]-T_i[\pi(j)]  \;\iff\; P_\pi[j] = (T_i)_\pi[j] \,.
    \end{equation*}
\end{lemma}
\begin{proof}
    The left-hand side of the arrow is the same as $P[\pi(j)]-P[j] = T_i[\pi(j)]-T_i[j]$, which by Definition~\ref{dfn:pdiffstr} is equivalent to the right-hand side of the arrow.
\end{proof}

\begin{lemma}
    \label{lem:TI-permup}
    Let $\pi$ be a permutation, $P$ a pattern and $T_i$ a text substring.
   \begin{equation*}
       \dHam(i)\leq k \,\implies\, \ham(P_\pi,(T_i)_{\pi}) \leq 2k \,.
    \end{equation*}
\end{lemma}
\begin{proof}
    Let $P$ and $T_i$ be such that $\dHam(i)\leq k$. By definition there exists an $\alpha$ such that the set $J=\Set{ j | P[j]+\alpha \neq T_i[j] }$ has size at most $k$. As $\pi$ is a permutation, there are at most $2k$ positions $j\in [m]$ such that either $j\in J$ or $\pi(j)\in J$. Therefore, for all (at least) $m-2k$ remaining positions $j'\in [m]$ we have that $P[j']+\alpha=T_i[j']$ and $P[\pi(j')]+\alpha=T_i[\pi(j')]$. For each such position $j'$, by rearranging the two equations it follows from Lemma~\ref{lem:TI-kmisiff} that $P_\pi[j'] = (T_i)_\pi[j']$. Thus, there are at most $2k$ mismatches between $P_\pi$ and $(T_i)_{\pi}$.
\end{proof}

A logical next step would be to attempt to find a permutation $\pi$ with the property that $\ham(P_\pi,(T_i)_{\pi}) \leq 2k$ implies that $\dHam(i) \leq k$ for all $P,T_i$. Unfortunately, Lemma~\ref{lem:notconv} below shows that no such permutation can exist. As Corollary~\ref{cor:notktight} states, this immediately implies that there is no permutation which is $k$-tight for all $P,T_i$. Instead we will select our permutation at random and show that we can obtain a permutation that is $k$-tight for a given $P,T_i$ with constant probability.

\begin{lemma}
    \label{lem:notconv}
    Let $\pi$ be any permutation and $6 \leq k < m/4$. There exists a pattern $P$ and text substring $T_i$ such that
    \begin{equation*}
        \dHam(i) > k \quad \text{and} \quad \ham(P_\pi,(T_i)_{\pi}) \leq 2k \,.
    \end{equation*}
\end{lemma}
\begin{proof}
    We define $P$ to be an $m$-length string of zeros. In order to define the $m$-length string $T_i$ we first introduce some notation.

    Let $k'=\lfloor k/2 \rfloor +1$. We identify a set of $k'$ locations $\ell_0,\ell_1,\ldots, \ell_{k'-1}\in [m]$ as follows. Location $\ell_0=0$. For $q\in\{1,\dots,k'-1\}$, location $\ell_q$ is the smallest position in $[m]$ that is not any of the preceding locations $\ell_0,\dots,\ell_{q-1}$ or any location that is mapped to or from by any of these preceding locations (under $\pi$). Formally, $\ell_q$ is the smallest location which is not in the set $L_q=\Set{ \ell_{q'},\, \pi(\ell_{q'}),\, \pi^{-1}(\ell_{q'}) | q'\in [q]}$. Observe that the set $L_q$ has size at most $3k' \leq 3(k/2+1) < 3k < m$ (since $k<m/4$), hence such a location always exists.

    We can now define $T_i$ as follows. For all $q\in [k']$, let $T_i[\ell_q]=1$ and $T_i[\pi(\ell_q)]=1$. At all other locations $j$, $T_i[j]=0$. Observe that by construction, the locations $\ell_0,\ldots, \ell_{k'-1}$ and $\pi(\ell_0),\ldots, \pi(\ell_{k'-1})$ are all distinct. Therefore, $T_i$ contains exactly $2k'$ ones and $m-2k'$ zeros. As $2k'\leq k+2< m/2$, more than half the locations have $T_i[j]=P[j]=0$, and therefore $\dHam(i)$ is minimised by the shift $\alpha=0$. Thus, $\dHam(i) = 2k' > k$.

    We proceed by showing that the alignment of $P_\pi$ and $(T_i)_\pi$ contains at least $m-3k'$ matches. There are $m-2k'$ locations $j$ in $T_i$ such that $T_i[j]=0$. Of these locations, at most $2k'$ have $T_i[\pi(j)]=1$. Therefore, there are at least $m-4k'$ locations $j$ such that $T_i[j]=T_i[\pi(j)]=0$. As $P[j]=P[\pi(j)]=0$, we have by Lemma~\ref{lem:TI-kmisiff} that $P_\pi[j]=(T_i)_\pi[j]$ at $m-4k'$ locations. Now consider locations $\ell_q$ for $q\in [k']$. By construction, $T_i[\ell_q]=T_i[\pi(\ell_q)]=1$ and therefore $P_\pi[\ell_q]=(T_i)_\pi[\ell_q]$ by Lemma~\ref{lem:TI-kmisiff}. This implies a further $k'$ matching locations. There are therefore at least at least $m-3k'$ matches or at most $3k'$ mismatches between $P_\pi$ and $(T_i)_\pi$. Since $3k' \leq 3(k/2 + 1) \leq 2k$ for all $k \geq 6$ we have that $\ham(P_\pi,(T_i)_\pi) \leq 2k$.
\end{proof}

\begin{corollary}
    \label{cor:notktight}
    Let $\pi$ be any permutation and $6 \leq k < m/4$. There exists a pattern $P$ and text substring $T_i$ for which $\pi$ is not $k$-tight.
\end{corollary}
\begin{proof}
Immediate from Definition~\ref{dfn:TI-ktight} and Lemma~\ref{lem:notconv}.
\end{proof}

\subsubsection{Random permutations} \label{sec:permutations}

We will choose a permutation uniformly at random  from a simple family of permutations. On first inspection, we could have chosen from the family of all permutations. We claim without proof that a permutation chosen uniformly at random from the family of all permutations is $k$-tight for any $P,T_i$ with constant probability. However, we must be able to efficiently compute  $\ham(P_\pi,(T_i)_{\pi})$ for all $i$ under our chosen permutation. The key problem being that in general $(T_i)_\pi$ is not easily obtained from $T$. As $i$ varies, $(T_i)_\pi$ could change drastically, even when $i$ is only incremented by one. Therefore we must be careful in selecting our family of permutations.

We will use the family of cyclic permutations, denoted $\Fam$ (for patterns of length $m$), defined as follows.

\begin{definition}
    \label{dfn:cperm}
    The set $\Fam$ contains the $m-1$ cyclic permutations $\pi_1,\pi_2,\ldots, \pi_{m-1}$, where
    \begin{equation*}
        \pi_q(j) = j + q \bmod m \,.
    \end{equation*}
\end{definition}

We now show in Lemma~\ref{lem:k-tight} that $\Fam$ has the desired property of $k$-tightness when $m > 6k^2$. There is a corner case when $k\in\{0,1\}$ which is easily solved in $O(n)$ time using our deterministic algorithm from Section~\ref{sec:detkmis}. For Lemma~\ref{lem:k-tight} we require that $k\geq 2$.

\begin{lemma}
    \label{lem:k-tight}
    Let $P$ be a pattern and $T_i$ a text substring. When $m > 6k^2$ and $k \geq 2$,
    \begin{equation*}
        \frac{\big|\Set{ \pi | \pi\in \Fam \text { is $k$-tight for $P,T_i$} } \big|}{|\Fam|} \geq \frac{1}{6} \,.
    \end{equation*}
\end{lemma}
\begin{proof}
    Let $\rho = \big|\Set{ \pi | \pi\in \Fam \text { is $k$-tight for $P,T_i$} } \big|/|\Fam|$. We will show that $\rho\geq 1/6$. Note that $|\Fam|=m-1$. We let $h=\dHam(i)$ be the minimal number of mismatches between $P$ and $T_i$, and $\alphah$ be the shift which minimises $\dHam(i)$.

    Assume first that $h \leq k$. By Lemma~\ref{lem:TI-permup} and Definition~\ref{dfn:TI-ktight} we have that that every $\pi\in\Fam$ is $k$-tight for $P,T_i$ and therefore $\rho=1$. Assume second that that $h>k$.  We split the proof into three cases:
    \begin{equation*}
        \textbf{Case 1. \;} k < h \leq 2k \qquad\quad
        \textbf{Case 2. \;} 2k < h \leq \frac{m}{3} \qquad\quad
        \textbf{Case 3. \;} \frac{m}{3} < h
    \end{equation*}

    First we introduce some notation. There are exactly $m-h$ positions $j$ where $\alphah+P[j]=T_i[j]$. We call such a position an $\alphah$-match. Similarly, any position with $\alpha+P[j]=T_i[j]$ for some $\alpha$ is called an $\alpha$-match. Positions which are not $\alpha$-matches are called $\alpha$-mismatches. Hence there are $h$ distinct $\alphah$-mismatches. We will refer to $\pi(j)$ as the position that $j$ is mapped to (by $\pi$).
    \medskip

    \noindent
    \textbf{Case 1 ($k < h \leq 2k$).}
    Let $j$ be an arbitrary $\alphah$-mismatch. Position $j$ is mapped to another  $\alphah$-mismatch in exactly $h-1$ distinct permutations of $\Fam$. This holds for each of the $h$ distinct $\alphah$-mismatches. Hence there are at most $(h-1)h$ permutations under which some $\alphah$-mismatch is mapped to another $\alphah$-mismatch. The remaining (at least) $(m-1)-(h-1)h$ permutations $\pi$ in $\Fam$ immediately have the following two properties:
    \begin{enumerate}
        \setlength{\itemsep}{0pt}
        \item[(i)] if position $j$ is an $\alphah$-mismatch then $\pi(j)$ is an $\alphah$-match;
        \item[(ii)] if position $\pi(j)$ is an $\alphah$-mismatch then position $j$ is an $\alphah$-match.
    \end{enumerate}
    There are $h$ positions $j$ with property~(i) and another (disjoint) $h$ positions $j$ with property~(ii). That is, for each $\alphah$-mismatch there are two positions $j$ that meet one of the two properties above. By Lemma~\ref{lem:TI-kmisiff}, each such $j$ implies that $P_\pi[j] \neq (T_i)_\pi[j]$. Therefore, in each of these $(m-1)-(h-1)h$ permutations $\pi$, $\ham(P_\pi,(T_i)_\pi) \geq 2h > 2k$ and so each such permutation is $k$-tight for $P,T_i$. By the assumption of Case~1, $h\leq 2k$, and the assumptions that $m \geq 6k^2$ and $k\geq 2$, we have that $(m-1)-(h-1)h \geq m - 4k^2 \geq m/3$. Thus, $\rho \geq (m/3)/(m-1) > 1/6.$
    \medskip

    \noindent
    \textbf{Case 2 ($2k < h \leq m/3$).}
    Let $K$ be an arbitrary set of $2k$ distinct $\alphah$-mismatches. For any permutation $\pi$, let
    \begin{equation*}
        K_\pi^{-1}= \Set{ j | \pi(j) \in K } \,.
    \end{equation*}
    We define
    \begin{equation*}
        \ham_{K}(P_\pi,(T_i)_\pi) =
            \big| \Set{ j | j\in (K \cup K_\pi^{-1}) \,\wedge\, P_\pi[j] \neq (T_i)_\pi[j] } \big|
    \end{equation*}
    to be the number of mismatch positions between $P_\pi$ and $(T_i)_\pi$ that are also in $K$ or $K_\pi^{-1}$. We now consider the total number of mismatches between $P_\pi$ and $(T_i)_\pi$ (that are in $K$ or $K_\pi^{-1}$) summed over all permutations in $\Fam$. Let
    \begin{equation*}
        H_K(P,T_i) = \sum_{\pi \in \Fam} \ham_K(P_\pi,(T_i)_\pi) \,.
    \end{equation*}

    Since $h\leq m/3$ by the assumption of Case~2, there are at least $2m/3$ $\alphah$-matches. A permutation $\pi$ that maps a position $j \in K$ to an $\alphah$-match creates a mismatch $P_{\pi}[j] \neq (T_i)_\pi[j]$ by Lemma~\ref{lem:TI-kmisiff} (as $j$ is an $\alphah$-mismatch). For a fixed $j\in K$, the number of permutations in $\Fam$ that map $j$ to an $\alphah$-match equals the number of $\alphah$-matches, which is at least $2m/3$. Thus, the set $K$ of $2k$ $\alphah$-mismatches contributes at least $2k \cdot (2m/3)$ to $H_K(P,T_i)$.

    Similarly, any position $j$ which is an $\alphah$-match creates a mismatch $P_{\pi}[j] \neq (T_i)_\pi[j]$ by Lemma~\ref{lem:TI-kmisiff} if it is mapped to an $\alphah$-mismatch in $K$. This occurs under exactly $2k$ permutations. Recall that any $j$ which is mapped to a position in $K$ under $\pi$ belongs to $K_\pi^{-1}$. Therefore, given that there are at least $2m/3$ $\alphah$-matches, the contribution is at least $2k \cdot (2m/3)$ further distinct mismatches to  $H_K(P,T_i)$.

    Summing up the previous two paragraphs, we have shown that $H_K(P,T_i) \geq (8/3)mk$.
    Each permutation $\pi$ that is not $k$-tight for $P,T_i$ has $\ham(P_\pi,(T_i)_\pi) \leq 2k$ (since $h>k$). Therefore, $m \cdot 2k$ is a generous upper bound on the number of mismatches across all permutations which are not $k$-tight. This leaves at least $(8/3)mk-2mk=(2/3)mk$ mismatches among the $k$-tight permutations of $\Fam$. Since $|K \cup K_\pi^{-1}| \leq 4k$, we have that $\ham_K(P_\pi,(T_i)_\pi) \leq 4k$ for any $\pi$, hence each permutation contributes at most $4k$ mismatches to $H_K(P,T_i)$. Therefore there are at least $(2/3)mk/(4k)=m/6$ distinct $k$-tight permutations. Thus, $\rho \geq (m/6)/(m-1) \geq 1/6$.
    \medskip

    \noindent
    \textbf{Case 3 ($m/3 < h$).}
    Similarly to Case~2, we consider the total number of mismatches between $P_\pi$ and $(T_i)_\pi$ summed over all permutations in $\Fam$. Let
    \begin{equation*}
        H(P,T_i) = \sum_{\pi \in \Fam} \ham(P_\pi,(T_i)_\pi) \,.
    \end{equation*}
    Since $h > {m/3}$ the number of $\alpha$-mismatches is more than $m/3$ for all $\alpha$. Fix an arbitrary position $j$ and choose an $\alpha$ such that $j$ is an $\alpha$-match. There are at least $m/3$ permutations $\pi$ in $\Fam$ that map position $j$ to an $\alpha$-mismatch. By Lemma~\ref{lem:TI-kmisiff}, $P_\pi[j] \neq (T_i)_\pi[j]$ for each of these permutations. Hence position $j$ will contribute with at least $m/3$ to $H(P,T_i)$. By considering all $m$ positions $j$, we have that $H(P,T_i) \geq m\cdot (m/3)$.

    Similarly to the reasoning in Case~2, each permutation $\pi$ that is not $k$-tight for $P,T_i$ has $\ham(P_\pi,(T_i)_\pi) \leq 2k$ (since $h>k$). Again, $m \cdot 2k$ is a generous upper bound on the number of mismatches across all permutations which are not $k$-tight. This leaves at least $m^2/3-2mk$ mismatches among the $k$-tight permutations of $\Fam$. As certainly $\ham(P_\pi,(T_i)_\pi) \leq m$, we have that there are at least $m/3-2k$ distinct $k$-tight permutations for $P,T_i$. Therefore,
    \begin{equation*}
        \rho \geq \frac{m/3-2k}{m-1} \geq \frac{k-1}{3k} \geq \frac{1}{6} \,,
    \end{equation*}
    where the second inequality follows from $m > 6k^2$ and the last inequality from $k \geq 2$, both assumptions in the statement of the lemma.
\end{proof}

\subsubsection{The algorithm}

Before describing the randomised algorithm we turn our attention to the problem of finding all positions $i\in [n-m+1]$ such that $\ham(P_{\pi},(T_i)_{\pi}) \leq 2k$ under an arbitrary cyclic permutation $\pi \in \Fam$. We will describe a simple deterministic algorithm that computes $\ham(P_\pi,(T_i)_{\pi})$ by reduction to the conventional $k$-mismatch problem.

Let $\pi_q\in \Fam$ be a fixed but arbitrary permutation ($q\in [1,\dots,m-1]$). Recall that $\pi_q(j) = j + q \bmod m$. We define
\begin{align*}
    P^+_q &= P_{\pi_q}[0\upto (m-q-1)] \,, \\
    P^-_q &= P_{\pi_q}[(m-q)\upto (m-1)] \,.
\end{align*}
Thus, $P_{\pi_q} = P^+_q\concat P^-_q$. We have $|P^+_q|=m-q$ and $|P^-_q|=q$. Now define $T^+_q$ and $T^-_q$ such that
\begin{align*}
    T^+_q[j] &= T[j+q] - T[j] \,, \\
    T^-_q[j] &= T[j+q-m] - T[j] \,,
\end{align*}
for all $j\in [n]$ (except those that take the indices ``out of range''). Observe that
\begin{equation*}
    (T_i)_{\pi_q} = T^+_q[i\upto (i+m-q-1)] \;\concat\; T^-_q[(i+m-q)\upto (i+m-1)] \,,
\end{equation*}
where the first substring has length $m-q$ and the second substring has length $q$. From these definitions it now follows directly that
\begin{align}
    \label{eq:hamconcat}
    \ham(P_{\pi_q},(T_i)_{\pi_q}) = \hspace{1em}
        &\,\ham\big(P^+_q,T^+_q[i\upto (i+m-q-1)]\big) \\
        + &\,\ham\big(P^-_q,T^-_q[(i+m-q)\upto (i+m-1)]\big) \,. \notag
\end{align}
Thus, in order to determine which positions $i$ have $\ham(P_{\pi_q},(T_i)_{\pi_q})\leq 2k$, we first construct $P^+_q$, $P^-_q$, $T^+_q$ and $T^-_q$, and then we run a standard $2k$-mismatch algorithm on the pairs $(P^+_q,T^+_q)$ and $(P^-_q,T^-_q)$ and use the previous formula.

We can now finally give an overview of our randomised algorithm for the \skDecision problem. The steps are described in Algorithm~\ref{alg:random}. The overall running time and proof of correctness is given in Theorem~\ref{thm:rankmis} below. The algorithm makes one-sided errors and outputs a false match (incorrectly reports $\dHam(i) \leq k$) with constant probability per alignment. As we will see in the proof of Theorem~\ref{thm:rankmis}, by running the algorithm a logarithmic number of times drastically reduces the probability of an error occurring at one or more alignments.

\begin{algorithm}[t]
    \caption{Overview of randomised solution to \skDecision.
        \label{alg:random}}
    \begin{enumerate}
        \item Pick a cyclic permutation $\pi_q \in \Fam$ uniformly at random.
        \item Construct the strings $P^+_q$, $P^-_q$, $T^+_q$ and $T^-_q$.
        \item Run a $2k$-mismatch algorithm on the pairs $(P^+_q,T^+_q)$ and $(P^-_q,T^-_q)$ as a black box.
        \item Using the results from Step~3 and Equation~(\ref{eq:hamconcat}), compute $\ham(P_{\pi_q},(T_i)_{\pi_q})$ for all~$i$.
        \item Any alignment $i$ with $\ham(P_{\pi_q},(T_i)_{\pi_q}) \leq 2k$ is declared to have $\dHam(i) \leq k$.
    \end{enumerate}
    \vspace{-8pt}
\end{algorithm}

\begin{theorem}
    \label{thm:rankmis}
    For any choice of constant $c$, \skDecision can be solved randomised in $O(cn \sqrt{k \log k}\log n)$ (deterministic) time when $k< \sqrt{m/6}$. The algorithm makes only false-positive errors (incorrectly declares the Hamming distance is at most~$k$). With probability at least $1-1/n^c$, the algorithm is correct at every alignment.
\end{theorem}
\begin{proof}
    As discussed in Section~\ref{sec:permutations}, if $k\in \{0,1\}$ then we can use the deterministic algorithm from Section~\ref{sec:detkmis} and achieve time complexity of $O(n)$ and no errors. Therefore, we focus on the case that $k\geq 2$.

    We first consider correctness. It follows from the discussion above that Algorithm~\ref{alg:random} does indeed determine, for every alignment $i$, whether $\ham(P_{\pi_q},(T_i)_{\pi_q}) \leq 2k$. We first show that
    \begin{enumerate}
        \item[(i)] $\ham(P_{\pi_q},(T_i)_{\pi_q}) \leq 2k$ when $\dHam(i) \leq k$;
        \item[(ii)] the probability that $\ham(P_{\pi_q},(T_i)_{\pi_q}) \leq 2k$ when $\dHam(i) > k$  is at most $5/6$.
    \end{enumerate}

    By Lemma~\ref{lem:TI-permup} we have that if $\dHam(i) \leq k$ then $\ham(P_{\pi_q},(T_i)_{\pi_q}) \leq 2k$. This proves property~(i). By Definition~\ref{dfn:TI-ktight}, $\ham(P_{\pi_q},(T_i)_{\pi_q}) > 2k$ if $\dHam(i) > k$ for all permutations $\pi_q$ that are $k$-tight for $P,T_i$. The permutation $\pi_q$ is selected uniformly at random from $\Fam$ in Step~1, hence by Lemma~\ref{lem:k-tight} it is $k$-tight for $P,T_i$ with probability at least $1/6$. This proves property~(ii). Note that we can apply Lemma~\ref{lem:k-tight} since we have assumed that $m>6k^2$ and $k\geq 2$.

    As Algorithm~\ref{alg:random} only makes false-positive errors, we can amplify the probability of giving correct outputs by repeating the algorithm. We repeat it $4(c+1) \lceil \log n \rceil$ times, where $c$ is a constant, and output any alignment which is reported by all repeats. More precisely, let $i$ be some alignment such that $\dHam(i) > k$. The probability that one run of Algorithm~\ref{alg:random} incorrectly reports position $i$ as a match is at most $5/6$. Thus, the probability that all runs output $i$ as a match is at most
    \begin{equation*}
        (5/6)^{4(c+1)\lceil \log n \rceil}<(1/2)^{(c+1)\log n}<1/n^{c+1} \,.
    \end{equation*}
    By the union bound over all positions $i$, the probability of the multi-run algorithm outputting a false match in at least one alignment is at most $n \cdot 1/n^{c+1}= 1/n^c$ as required.

    We now consider the time complexity of Algorithm~\ref{alg:random} (without amplification). Step~1 requires only constant time to pick a permutation at random. Step~2 requires $O(n)$ time by inspection of the definitions. Step~3 makes two calls to a $2k$-mismatch algorithm. For both calls the input is a pattern of length $O(m)$ and a text of length $O(n)$. Using the fastest known $k$-mismatch algorithm of Amir et al.~\cite{ALP:2004}, this step takes $O(n \sqrt{k \log k})$ time. Steps~4 and~5 require only scanning the output of Step~3 and therefore take $O(n)$ time. This gives a time complexity of $O(n \sqrt{k \log k})$ time. However, we repeat the algorithm $O(c\log n)$ times to reduce the error probability, hence $O(cn \sqrt{k \log k}\log n)$ is the total time complexity.
\end{proof}

\section{Discussion}\label{sec:discussion}

We have shown how to derive both new upper and lower bounds for a variety of pattern matching problems under polynomial transformations. In some cases we have improved on known results and in others introduced new problem definitions and solutions.  There remain however a number of open questions. First, we suspect that the true complexity of \LpolyWild is unresolved, particularly for higher polynomial transformations. For example, when $r=m$ there exists a straightforward $O(nm)$ time solution by considering the problem independently at each alignment.  It is also still uncertain if the normalised Hamming distance problem is \threeSUM-hard for polynomials of degree greater than one. For \skDecision, our fast randomised algorithm applies only when $k < \sqrt{m/6}$. However, our lower bound for the same problem applies to the case where we want to determine if the Hamming distance is at most $m-2$. This leaves a range of values of $k$ where the complexity is not yet determined.  It is also an interesting question whether our randomised solution can be efficiently modified to output the Hamming distance at each alignment rather than simply a decision about whether it is greater or less than $k$ or indeed if a new fast method can be found for this problem which will allow the presence of wildcards in the input.

\section{Acknowledgements}
MJ and BS are both supported by the EPSRC. The authors are grateful to Philip Bille for very helpful discussions on the topic of \threeSUM and related problems and their application to lower bounds for pattern matching problems.

\printbibliography

\end{document}